\documentclass[preprint, pra, amsmath]{revtex4}
\usepackage{enumerate}
\usepackage{graphicx}
\usepackage{amsmath, amsthm, amssymb}
\newtheorem{theorem}{Theorem}

\newtheorem{lemma}{Lemma}

\theoremstyle{remark}
\newtheorem*{remark}{Remark}
\usepackage{enumerate}
\usepackage{graphicx}
\usepackage{amsmath, amssymb}
\usepackage{graphics}
\begin{document}
\newcommand{\real}{\textrm{Re}\:}
\newcommand{\sto}{\stackrel{s}{\to}}
\newcommand{\Tr}{\textrm{Tr}\:}
\newcommand{\supp}{\textrm{supp}\:}
\newcommand{\wto}{\stackrel{w}{\to}}
\newcommand{\ssto}{\stackrel{s}{\to}}
\newcounter{foo}
\providecommand{\norm}[1]{\lVert#1\rVert}
\providecommand{\abs}[1]{\lvert#1\rvert}

\title{Zero Energy Bound States in Three--Particle Systems}

\author{Dmitry K. Gridnev}
\affiliation{FIAS, Ruth-Moufang-Stra{\ss}e 1, D--60438 Frankfurt am Main, Germany}
\altaffiliation[On leave from:  ]{ Institute of Physics, St. Petersburg State University, Ulyanovskaya 1, 198504 Russia}

\begin{abstract}
Under certain restrictions on pair--potentials it is proved that the eigenvalues in the three--particle system are absorbed at zero energy threshold if there is no negative energy bound states and zero energy resonances in particle pairs.
\end{abstract}

\maketitle


\section{Introduction}\label{sec:1}

We consider the $N$--particle Schr\"odinger operator
\begin{equation}\label{e111}
    H(\lambda) = H_0 - \lambda \sum_{1 \leq i<j \leq N} V_{ij}(r_i - r_j),
\end{equation}
where $\lambda > 0$ is a coupling constant, $H_0$ is a kinetic energy operator with the center of mass removed, $r_i \in \mathbb{R}^3$ are particle poaition vectors, the pair potentials are real (further restrictions on the potentials would be given later). Suppose that for $\lambda$  in the vicinity of some $\lambda_{cr} < \infty$ there is a bound state $\psi(\lambda) \in D(H_0)$ with the energy $E(\lambda) < \inf \sigma_{ess} (H(\lambda))$ and $E(\lambda) \to \inf \sigma_{ess}(H(\lambda_{cr}))$ when $\lambda \to \lambda_{cr}$. The question whether $E(\lambda_{cr}) \in \sigma_{pp} (H(\lambda_{cr}))$ was considered in various contexts in \cite{klaus,klaus2,simon,ostenhof,gest,prl} (the list of references is by far incomplete).

In \cite{karner}, Theorem~3.3,  it was claimed that if $V_{ij} \in C_0^\infty (\mathbb{R}^3)$, $V_{ij} \geq 0$, and none of the subsystems has negative energy bound states or zero energy resonances, then there exists $\psi(\lambda_{cr}) \in D(H) = D(H_0), \psi(\lambda_{cr}) \neq 0$ such that $H(\lambda_{cr}) \psi(\lambda_{cr})= 0$. Unfortunately, the proof in \cite{karner} contains a mistake. In Eq.~53 of \cite{karner} the mixed term containing first order derivatives is erroneously omitted, which makes the results of Ref.~35 in \cite{karner} concerning the fall off of the wave function inapplicable. And it is not immediately clear how the arising hurdle can be overcome. Here we prove the result stated by Karner for $N=3$ with a different method and for a larger class of potentials (Theorem~\ref{th:2} of this paper). In the next publication \cite{submit} we demonstrate that the condition on the absence of zero energy resonances in particle pairs is essential, {\it i.e.} under certain conditions appearance of a zero energy resonance in one of the two--body subsystems makes the statement false.

Note, that the 3--body case differs essentially from the 2--body case, where under similar restrictions on pair potentials the zero--energy \textbf{ground} state can never be a bound state \cite{klaus,prl}. The conclusion that a zero energy resonance in the three--body system is in fact a bound state is unexpected and has far reaching physical consequences, which concern the size of a system in its ground state (we ignore the particle statistics here). In the two--body case the size of the system in the ground state can be made infinite by tuning, for example, the coupling constant so that the bound state with negative energy approaches the zero energy threshold \cite{prl}. In the three body case the size of the system remains finite, given that in the course of tuning the coupling constants of two--body subsystems stay away from critical values, at which the two--body zero energy resonances appear. To underline the connection with the size of the system we formulate the proofs in terms of spreading and non--spreading sequences of bound states. The result has applications in the physics of halo nuclei \cite{vaagen}, molecular physics \cite{fedorov} and Efimov states \cite{efimov}.

The paper is organized as follows. In Sec.~\ref{sec:2} we use the ideas of Zhislin \cite{zhislin} to set up the framework for the analysis of eigenvalue absorption in connection with the spreading of sequences of wave functions. Here we prefer to maintain generality and do not restrict ourselves to $N=3$. In Sec.~\ref{sec:3} we consider the 3--body case and employ the equations of Faddeev type to prove Theorem~\ref{th:2}, which is the main result of the paper.

\section{Spreading and Bound States at Threshold}\label{sec:2}

The main result of this section (Theorem~\ref{th:1}) appears implicitly in \cite{zhislin}, where Zhislin considers minimizing sequences of the energy functional in Sobolev spaces. For our purposes it is more useful to consider sequences of eigenstates and use an approach in the spirit of \cite{simon}.

Consider the $N$-particle Hamiltonian, which depends on a parameter
\begin{gather}
    H(\lambda) = H_0 + V(\lambda) ,  \label{xc31}  \\
V(\lambda) =  \sum_{1 \leq i<j \leq N} V_{ij} (\lambda; r_i - r_j ) \label{:xc31},
\end{gather}
where $H_0$ is the kinetic energy operator with the center of mass removed, $V_{ij}$ are pair potentials and $r_i \in \mathbb{R}^3$ are position vectors. For the parameter $\lambda$ we assume that $\lambda \in \mathbb{R}$ (this is done for clarity, in fact, $\lambda$ can take values in a topological space). We impose the following set of restrictions.
\begin{list}{R\arabic{foo}}
{\usecounter{foo}
    \setlength{\rightmargin}{\leftmargin}}
\item $H(\lambda)$ is defined for an infinite sequence of parameter values $\lambda_1, \lambda_2, \ldots $ and $\lambda_{cr}$, where $\lim_{n\to \infty} \lambda_n = \lambda_{cr}$.

\item $|V_{ij} (\lambda; y )| \leq F (y)$ for all $\lambda$ defined in R1, where $V_{ij}, F \in L^2 (\mathbb{R}^3) + L_\infty^\infty (\mathbb{R}^3)$.

\item $\forall f \in C^\infty_0 (\mathbb{R}^{3N-3} )\colon  \lim_{\lambda_n \to \lambda_{cr}} \bigl\| \bigl[ V(\lambda_n) - V(\lambda_{cr}) \bigr] f \bigr\| =
0$.
\end{list}
The symbol $L_\infty^\infty$ denotes bounded Borel functions going to zero at infinity. By R2 $H(\lambda)$ is self-adjoint on $D(H_0 )$ \cite{reed}.

The bottom of the essential spectrum is denoted as
\begin{equation}\label{sigmaess}
  E_{thr} (\lambda) := \inf \sigma_{ess} (H(\lambda)) .
\end{equation}

The set of requirements on the system continues as follows
\begin{list}{R\arabic{foo}}
{\usecounter{foo}
    \setlength{\rightmargin}{\leftmargin}}
\setcounter{foo}{3}
\item
for all $ \lambda_n $ there is $E(\lambda_n) \in \mathbb{R}, \psi(\lambda_n) \in D(H_0)$ such that $H(\lambda_n) \psi(\lambda_n) = E(\lambda_n) \psi(\lambda_n)$, where $\| \psi(\lambda_n) \| = 1$ and $E(\lambda_n) < E_{thr}(\lambda_n)$.

\item
$\lim_{\lambda_n \to \lambda_{cr}} E(\lambda_n) = \lim_{\lambda_n \to \lambda_{cr}} E_{thr}(\lambda_n) = E_{thr} (\lambda_{cr})$ . 
\end{list}

The requirements R4-5 say that for each $n$ the system has a level below the continuum and for $\lambda_n \to \lambda_{cr}$ the energy of this level approaches the bottom of the continuous
spectrum.

In the proofs we shall use the term ``spreading sequence'', which is due to Zhislin \cite{zhislin}. The sequence of functions $f_n (x) \in L^2 (\mathbb{R}^d)$
\textbf{spreads} if there is $a>0$ such that $\limsup_{n \to \infty} \|  \chi_{\{x||x| > R\}} f_n \| > a$ for all $R>0$. (the notation $\chi_{\Omega}$  always means the characteristic function of the set $\Omega$). The sequence $f_n$ is
\textbf{totally spreading}  if $\lim_{n \to \infty} \|  \chi_{\{x||x| \leq R\}} f_n \| = 0$ for all $R >0$.

\begin{lemma}\label{lem:1}
Let $H(\lambda)$ be a Hamiltonian satisfying R1-5. Then
\begin{equation}
\sup_n \| H_0 \psi(\lambda_n)\| < \infty .
\end{equation}
\end{lemma}

\begin{proof}
The statement represents a
well-known fact, see e. g.  \cite{zhislin} but for
completeness we give the proof right here.
The Shr\"odinger equation $H_0 \psi (\lambda_n)= -V(\lambda_n)\psi (\lambda_n) + E(\lambda_n)
\psi(\lambda_n)$ gives the bound $\| H_0 \psi (\lambda_n) \| \leq \| V(\lambda_n
)\psi (\lambda_n) \| + O(1)$.
By R2 $|V_{ij}| \leq F_{ij}$, where for a shorter notation we denote $F_{ij} := F(x_i - x_j)$. Using
that $F_{ij}$ is $H_0$ bounded \cite{reed}
with a relative bound 0 we obtain
\begin{gather}
    \|V(\lambda_n) \psi (\lambda_n)\| = \bigl\| \sum_{i<j} V_{ij} (\lambda_n; x_i - x_j )  \psi (\lambda_n) \bigr\|
    \leq \frac{N(N-1)}2 \bigl\| F_{ij}  \psi (\lambda_n) \bigr\| \leq \label{kuka}\\
a \|H_0 \psi (\lambda_n)\| + b \leq
    a \| V(\lambda_n)\psi (\lambda_n)\| + O(1) , \label{kuka:1}
\end{gather}
where $a, b >0$ are constants and $a$ can be
chosen as small as pleased. Setting $a = 1/2$ and dividing
(\ref{kuka})--(\ref{kuka:1}) by $\| V(\lambda_n)
\psi(\lambda_n) \|$ we find that $\| V(\lambda_n) \psi(\lambda_n) \| $, respectively $\| H_0 \psi (\lambda_n)\| $ must be uniformly bounded.  \end{proof}

The following theorem illustrates the connection between non-spreading and bound states at threshold.

\begin{theorem}[Zhislin]\label{th:1}
Let $H(\lambda)$ satisfy R1-5. If the sequence  $\psi (\lambda_n)$ does not totally spread then $H(\lambda_{cr})$ has a  bound state at threshold $\psi_{cr} \in D(H_0)$, that is
\begin{equation}\label{xc11}
H (\lambda_{cr}) \psi_{cr} = E_{thr}(\lambda_{cr}) \psi_{cr} ,
\end{equation}.
\end{theorem}
For the proof we need a couple of technical Lemmas.

\begin{lemma}\label{lem:2}
Suppose $f_n \in D(H_0)$ is such that $\sup_n \| H_0 f_n\| < \infty$ and $f_n \wto f_0$. Then (a) $f_0 \in D(H_0)$; (b) for any operator $A$, which is relatively $H_0$ compact $\| A ( f_n - f_0 ) \| \to 0$.
\end{lemma}
\begin{proof}
First, let us prove that the sequence $H_0 f_n $ is weakly convergent. A proof is by contradiction. Suppose $H_0 f_n $ has two weak limit points, {\it i.e.} there exist $f'_k, f''_k$, which are subsequences of $f_n$ and for which $H_0 f'_k \wto \phi_1$ and $H_0 f''_k \wto \phi_2$, where $\phi_{1,2} \in L^2$ and $\phi_1 \neq \phi_2$. On one hand, because $\phi_1 \neq \phi_2$ and $D(H_0)$ is dense in $L^2$ there is $g \in D(H_0)$ such that $(\phi_1 - \phi_2 , g) \neq 0$. On the other hand, using that $f'_k \wto f_0$ and  $f''_k \wto f_0$ we get
\begin{equation}\label{:labi:}
    (\phi_1 - \phi_2 , g) = \lim_{k \to \infty} \left[ \bigl( H_0 (f'_k - f''_k), g\bigr)\right] =
    \lim_{k \to \infty} \left[ \bigl( (f'_k - f''_k), H_0 g\bigr)\right] = 0,
\end{equation}
a contradiction. Hence, $H_0 f_n \wto G$, where $G \in L^2$. $\forall f \in D(H_0)$ by self-adjointness of $H_0$ we obtain $(H_0 f, f_0) = \lim_{n \to \infty} (H_0 f, f_n) = (f, G)$. Thus $f_0 \in D(H_0)$ and $G = H_0 f_0$, which proves (a). To prove (b) note that $(H_0 + 1)(f_n - f_0 ) \wto 0$. Using that compact operators acting on weakly convergent sequences make them converge in norm we get
\begin{equation}\label{fck}
A(f_n - f_0 )  = A(H_0 + 1)^{-1} (H_0 + 1) (f_n - f_0 )\to 0 ,
\end{equation}
since $A(H_0 + 1)^{-1}$ is compact by condition of the lemma.   \end{proof}

\begin{lemma}\label{lem:3}
Suppose $f_n \in D(H_0)$ is such that $\sup_n \| H_0 f_n\| < \infty$ and $f_n \wto f_0$.
Then (a) if $f_n$ does not spread then $f_n \to f_0$ in norm; (b) if $f_n$ does not totally spread then $f_0 \neq 0$.
\end{lemma}
\begin{proof}
Let us start with (a).
Because $f_n$ does not spread it is enough to show that $\| \chi_{\{x| |x| \leq R\}} (f_n - f_0) \| \to 0$ for all $R$ in norm. And this is true because $\chi_{\{x| |x| \leq R\}}$ is relatively $H_0$ compact \cite{reed,teschl} and Lemma~\ref{lem:2} applies. To prove (b) let us assume by contradiction that $f_n \wto 0$. Using the same arguments we get that $\| \chi_{\{x| |x| \leq R\}} f_n \| \to 0$ for all $R$. But this would mean that $f_n$ totally spreads contrary to the condition of the Lemma.   \end{proof}

\begin{proof}[Proof of Theorem~\ref{th:1}]
Because $\psi(\lambda_n)$ does not totally spread there are $a, R > 0$ and a
subsequence $\lambda_k$ such that $\| \chi_{\{x| |x| < R\}} \psi(\lambda_k) \| > a$. From this subsequence by the
Banach-Alaoglu theorem we choose a weakly convergent sub/subsequence
(for which we keep the notation $\psi(\lambda_k)$) such that $\psi(\lambda_k) \wto
\psi_{cr}$, where $\psi_{cr} \in D(H_0)$ by Lemma~\ref{lem:2}. The
sub/subsequence $\psi(\lambda_k)$ does not totally spread and is weakly
convergent, hence,  by Lemma~\ref{lem:3}(b) $\psi_{cr} \neq 0$. For any $f \in
C_{0}^\infty$ we have
\begin{gather}
    \Bigl([H(\lambda_{cr}) - E_{thr}(\lambda_{cr})]f , \psi_{cr} \Bigr) = \lim_{\lambda_n \to \lambda_{cr}} \Bigl( [H(\lambda_{cr}) -E_{thr}(\lambda_n)]f , \psi (\lambda_n) \Bigr) = \\
\lim_{\lambda_n \to \lambda_{cr}} \Bigl( \bigl[ H(\lambda_n) - (V(\lambda_n)-V(\lambda_{cr}) )  -E_{thr}(\lambda_n) \bigr] f ,\psi (\lambda_n) \Bigr)  = \\
\lim_{\lambda_n \to \lambda_{cr}} \Bigl\{ \bigl[E(\lambda_n) -E_{thr}(\lambda_n) \bigr] \Bigl( f ,\psi (\lambda_n) \Bigr)  -
\Bigl([V(\lambda_n)-V(\lambda_{cr}) ]f ,\psi (\lambda_n)\Bigr) \Bigr\} = 0,
\end{gather}
where in the last equation we have used R3,5.
Summarizing, for all $f \in C_{0}^\infty$ we have
\begin{equation}\label{xc12}
    \left(\bigl[H(\lambda_{cr}) - E_{thr}(\lambda_{cr})\bigr]f , \psi_{cr} \right) =  \left(f , \bigl[H(\lambda_{cr}) - E_{thr}(\lambda_{cr})\bigr]\psi_{cr} \right) = 0,
\end{equation}
meaning that (\ref{xc11}) holds.  \end{proof}

The following Lemmas will be needed in the next Section.
\begin{lemma}\label{lem:4}
 A uniformly norm--bounded sequence of functions $f_n \in L^2 (\mathbb{R}^n)$, where every weakly converging subsequence converges also in norm, does not spread.
\end{lemma}
\begin{proof}
By contradiction, let us assume that $f_n$ spreads. Then it is
possible to extract a subsequence $g_k = f_{n_k}$ with the
property $\| \chi_{\{x| |x| \geq k \}} g_k \| > a$, where $a > 0$
is a constant. On one hand, it is easy to see that $g_k $ with
this property has no subsequences that converge in
norm. On the other hand, by the Banach-Alaoglu theorem $g_k$ must have at least one  weakly converging subsequence, which is norm--convergent by condition of the lemma.  \end{proof}

\begin{lemma}\label{lem:5}
Suppose $g \in C (\mathbb{R}^{3N-3})$ has the property that $|g|\leq 1$ and $g=0$ if $|r_i - r_j | < \delta |x|$, where $\delta$ is a constant. Then the operator $gF(r_i - r_j)$ is relatively $H_0$ compact.
\end{lemma}
\begin{proof}
It suffices to consider the case $F\in L^2 (\mathbb{R}^3)$ (the case $F \in L_\infty^\infty (\mathbb{R}^3)$ trivially follows from Lemma~7.11 in \cite{teschl}). For $k=1,2,\ldots$ we can write
\begin{equation}\label{muta}
gF_{ij}(H_0 + 1)^{-1} = \chi_{\{x|\: |r_i - r_j | <k\}}gF_{ij}(H_0 + 1)^{-1} + \chi_{\{x|\: |r_i - r_j | \geq k\}}gF_{ij}(H_0 + 1)^{-1} .
\end{equation}
The first operator on the rhs is compact (Lemma~7.11 in \cite{teschl}). We need to show that the second one goes to zero in norm when $k \to \infty$ (in this case the operator on the lhs is compact as a norm-limit of compact operators). The following integral estimate of the square of its norm is trivial
\begin{equation}\label{muta2}
\bigl\| \chi_{\{x|\: |r_i - r_j | \geq k\}}gF_{ij}(H_0 + 1)^{-1}\bigr\|^2 \leq \frac 1{(4\pi)^2}
\int_{|r|\geq k} d^3 r\;  |F(r)|^2 \int d^3 r' \; \frac{e^{-2|r'|}}{|r'|^2} .
\end{equation}
Because $F\in L^2 (\mathbb{R}^3)$ the rhs goes to zero as $k \to \infty$. \end{proof}

\section{The Case of Three Particles}\label{sec:3}

We apply the framework of Sec.~\ref{sec:2} to the system of three particles with non-positive potentials. The case $N>3$ and potentials taking both signs would be considered elsewhere. For simplicity we take the parameter $\lambda > 0$ as a coupling constant of the interaction (see \cite{klaus,klaus2})
\begin{gather}
    H(\lambda) = H_0 - \lambda V  , \label{xc31aa} \\
V =  \sum_{1 \leq i<j \leq 3} V_{ij} (r_i - r_j ) . \label{:xc31aa}
\end{gather}
We shall need the following additional requirements
\begin{list}{R\arabic{foo}}
{\usecounter{foo}
    \setlength{\rightmargin}{\leftmargin}}
\setcounter{foo}{5}
\item
$V_{ij} \geq 0$ and $\lambda V_{ij} (y) \leq F(y)$, where $F \in L^2 (\mathbb{R}^3) \cap L^1 (\mathbb{R}^3) $ and $\lambda$ takes values as defined in R1.
\item
There exists $\epsilon > 0$ such that $H_0 - (\lambda + \epsilon )V_{ij} \geq 0$ for all $\lambda$ defined in R1 and all pair potentials $V_{ij}$.
\end{list}
Requirement R7 means that the two--particle subsystems have no bound states with negative energy and no resonances at zero energy. This results in $E_{thr}(\lambda) = 0$. Our aim is to prove
\begin{theorem}\label{th:2}
Suppose $H(\lambda)$ defined in (\ref{xc31aa})--(\ref{:xc31aa}) satisfies R1, R4-7. Then for $n \to \infty$ the sequence $\psi_n$ does not spread and there exists a bound state at threshold $\psi_{cr} \in D(H_0)$ such that $H(\lambda_{cr})\psi_{cr} = 0$.
\end{theorem}
We shall defer the proof, which boils down to the construction of Faddeev equations \cite{faddeev}, see also \cite{sobolev,yafaev}, to the end of the section. Let us introduce an analytic operator function $B_{ij} (z)$ for each pair of particles $(ij)$. We shall construct $B_{12}$ and the other two operators are constructed similarly. We use Jacobi coordinates \cite{greiner} $x = [\sqrt{2 \mu_{12}}/\hbar](r_2 - r_1)$ and $y = [\sqrt{2 M_{12}}/\hbar](r_3 - m_1/(m_1+m_2) r_1 - m_2/(m_1+m_2) r_2)$, where $\mu_{ij} = m_i m_j /(m_i + m_j)$ and $M_{ij} = (m_i + m_j)m_l / (m_i + m_j + m_l)$ are reduced masses (the indices $i,j,l$ are all different). These coordinates make the kinetic energy operator take the form
\begin{equation}\label{ay4}
    H_0 = - \Delta_x - \Delta_y .
\end{equation}
Let $\mathcal{F}_{12}$ denote the partial Fourier transform in $L^2(\mathbb{R}^6)$ acting as follows
\begin{equation}\label{ay5}
\hat f(x,p_y)  =  \mathcal{F}_{12} f(x,y) = \frac 1{(2 \pi )^{3/2}} \int d^3 y \; e^{-ip_y \cdot \; y} f(x,y) .
\end{equation}
Then $B_{12}(z)$ is defined through
\begin{equation}\label{ay6}
B_{12}(z) = 1 + z  + \mathcal{F}^{-1}_{12} t(p_y) \mathcal{F}_{12},
\end{equation}
where
\begin{equation}\label{ay669}
t(p_y) = (\sqrt{|p_y|} - 1)\chi_{\{p_y | \; |p_y| \leq 1\}}.
\end{equation}
Similarly, using other Jacobi coordinates one defines $B_{ij}(z)$  and $\mathcal{F}_{ij}(z)$ for all particle pairs.
Note that $B_{ij}(z)$ and $B^{-1}_{ij}(z)$ are analytic on $\real z > 0$.

\begin{lemma}\label{lem:6}
The operator function in $L^2 (\mathbb{R}^6)$
\begin{equation}\label{ya23}
  \mathcal{A}_{ij} (z) = (H_0 + z^2)^{-1} V_{ij}^{1/2} B_{ij} (z)
\end{equation}
is uniformly bounded for $z \in (0,1]$, and strongly continuous for $z \to +0$.
\end{lemma}
\begin{proof}
We take the case when $(ij) = (12)$, other indices are treated similarly. Instead of $\mathcal{A}_{12} (z) $ we consider  $\mathcal{F}_{12}  \mathcal{A}_{12} (z) \mathcal{F}^{-1}_{12}   $. We take $z \in (0,1)$ and split the operator
\begin{equation}\label{redew1}
    \mathcal{F}_{12}  \mathcal{A}_{12} (z) \mathcal{F}^{-1}_{12}   = K_1 (z) + K_2 (z) ,
\end{equation}
where
\begin{gather}
K_1 (z) = (- \Delta_x + p_y^2 + z^2)^{-1} V^{1/2}_{12}(\alpha x) [t(p_y) + 1]  , \label{redew2} \\
K_2 (z) = (- \Delta_x + p_y^2 + z^2)^{-1} V^{1/2}_{12}(\alpha x) z     \label{redew3}
\end{gather}
are integral operators acting on $\phi (x, p_y)\in L^2(\mathbb{R}^6)$ as
\begin{gather}
K_1 (z)\phi  = \frac 1{4 \pi} \int d^3 x'  \frac{e^{-\sqrt{p_y^2 + z^2} |x-x'|}}{|x-x'|}  V^{1/2}_{12}(\alpha x') [t(p_y) + 1] \phi (x', p_y) ,  \label{redew6} \\
K_2 (z)\phi  = \frac{z }{4 \pi} \int d^3 x'  \frac{e^{-\sqrt{p_y^2 + z^2} |x-x'|}}{|x-x'|}  V^{1/2}_{12}(\alpha x') \phi (x', p_y) . \label{redew7}
\end{gather}
The numerical coefficient $\alpha$ depends on masses $\alpha := \hbar /\sqrt{2\mu_{12}} $. Applying the Cauchy-Shwarz inequality we get
\begin{gather}
\bigl| K_1 (z) \phi \bigr|^2 \leq \int d^3 x'  \frac{e^{-2 |p_y| |x-x'|}}{|x-x'|^2} [t(p_y) +1]^2 V_{12} (\alpha x')  \times \int d^3 x' \bigl| \phi (x', p_y) \bigr|^2  ,  \label{redew4} \\
\bigl| K_2 (z)\phi \bigr|^2 \leq z^2  \int d^3 x'  \frac{e^{-2 z |x-x'|}}{|x-x'|^2} V_{12} (\alpha x')  \times \int d^3 x' \bigl| \phi (x', p_y) \bigr|^2 \label{redew5} ,
\end{gather}
where we have used $ z \in (0,1]$. Integrating (\ref{redew4}) and (\ref{redew5}) over $x$ leads to
\begin{gather}
\int d^3 x \bigl| K_1 (z) \phi \bigr|^2 \leq cc'c''\Bigl[\int d^3 x' \bigl| \phi (x', p_y) \bigr|^2 \Bigr] , \label{wadi1}\\
\int d^3 x \bigl| K_2 (z) \phi \bigr|^2 \leq cc'\Bigl[\int d^3 x' \bigl| \phi (x', p_y) \bigr|^2 \Bigr] , \label{wadi2}
\end{gather}
where $c, c', c''$ are the following finite constants
\begin{gather}
    c = \int d^3 x'  V_{12} (\alpha x') , \label{aijacts9} \\
    c' = \int d^3 x \frac{e^{-2 |x|}}{|x|^2}  , \label{aijacts10} \\
    c'' = \sup_{p_y \in \mathbb{R}^3} [t(p_y) + 1 ]^2 / |p_y| . \label{aijacts4}
\end{gather}
Integrating (\ref{wadi1})--(\ref{wadi2}) over $p_y$ gives that $K_{1,2} (z)$ is uniformly norm--bounded for $z \in (0,1]$. The strong continuity for $z \to +0$ follows from (\ref{redew6})--(\ref{redew7}) by the dominated convergence theorem. \end{proof}

It is convenient to introduce the notation
\begin{equation}\label{ay2}
    \mathcal{C}_{ik;jm}(z) = V^{1/2}_{ik} (H_0 + z^2)^{-1} V^{1/2}_{jm} .
\end{equation}
We shall need the following
\begin{lemma}\label{lem:7}
Suppose R1, R4-7 are satisfied and $k^2_n := -E(\lambda_n)$. Then the operators
\begin{equation}\label{ay2222}
    \mathcal{R}_{ij}(\lambda_n) = [1 - \lambda_n \mathcal{C}_{ij;ij}(k_n)]^{-1}
\end{equation}
are uniformly bounded for all $n$ and converge in norm when $n \to \infty$.
\end{lemma}
\begin{proof}
The operators $\mathcal{C}_{ij;ij}(z)$ are uniformly bounded for $z>0$ and converge in norm $\mathcal{C}_{ij;ij}(z) \to \mathcal{C}_{ij;ij}(0)$ for $z \to +0$ (this follows from writing out the kernel $\mathcal{F}_{ij} \mathcal{C}_{ij;ij}(k_n) \mathcal{F}^{-1}_{ij}$ explicitly, like in the proof of Lemma~\ref{lem:6}, and checking that $\| \mathcal{C}_{ij;ij}(z) - \mathcal{C}_{ij;ij}(0)\| \to 0$). From R7 and Birman--Schwinger principle \cite{reed,klaus} $\| \lambda_n C_{ij;ij} (k_n^2) \| < 1- \varepsilon$, where $\varepsilon > 0$ is a constant. From expanding (\ref{ay2222}) in von Neumann series it follows that $\mathcal{R}_{ij}(\lambda_n)$ converges in norm. \end{proof}

\begin{lemma}\label{lem:8}
For $(ik) \neq (jm)$ the operator function $B^{-1}_{ik} (z) \mathcal{C}_{ik;jm} (z) $
is uniformly norm--bounded for $z \in (0,1]$ and strongly continuous for $z \to +0$.
\end{lemma}
\begin{proof}
We focus on $B^{-1}_{12} (z) \mathcal{C}_{12;23} (z) $, the other indices are treated similarly. Let us show that $\mathcal{F}_{12} B^{-1}_{12} (z) \mathcal{C}_{12;23} (z)\mathcal{F}_{12}^{-1}$ is uniformly bounded for $z \in (0,1]$.
\begin{equation}\label{splitti}
    \mathcal{F}_{12} B^{-1}_{12} (z) \mathcal{C}_{12;23} (z)\mathcal{F}_{12}^{-1} = K_1 (z)+ K_2(z),
\end{equation}
where
\begin{gather}
 K_1 (z)= \frac 1{z + 1} \mathcal{F}_{12} \mathcal{C}_{12;23} (z)\mathcal{F}_{12}^{-1} ,  \label{k1}\\
  K_2 (z)= \mathcal{F}_{12} \Bigl(B^{-1}_{12}(z) - \frac 1{z + 1}  \Bigr)\mathcal{C}_{12;23} (z)\mathcal{F}_{12}^{-1} .  \label{k2}
\end{gather}
$\mathcal{C}_{12;23} (z)$ is uniformly bounded for $z \in (0,1]$. Indeed, $\mathcal{C}_{12;23} (z)$ is a product of $V^{1/2}_{12}(H_0 + z^2)^{-1/2}$ and $(H_0 + z^2)^{-1/2} V^{1/2}_{23}$, which gives $\| \mathcal{C}_{12;23} (z)\| \leq \| \mathcal{C}_{12;12} (z)\|^{1/2} \| \mathcal{C}_{23;23} (z)\|^{1/2} \leq \| \mathcal{C}_{12;12} (0)\|^{1/2} \| \mathcal{C}_{23;23} (0)\|^{1/2} $, where each norm is bounded by the Birman--Schwinger principle. Thus $K_1 (z)$ is uniformly norm--bounded for $z \in (0,1]$.

Below we prove that the Hilbert-Schmidt norm of $K_2 (z)$ is bounded for $z\in (0,1]$. Let us first consider the Fourier transformed interaction term $\mathcal{F}_{12} V^{1/2}_{23} \mathcal{F}_{12}^{-1}$. In  Jacobi coordinates the interaction term has the form $V^{1/2}_{23} = V^{1/2}_{23}(\beta x + \gamma y)$, where $\beta $ and $\gamma \neq 0$ are real constants depending on masses $\beta = -m_2 \hbar / ((m_1 + m_2)\sqrt{2m_{12}})$ and $\gamma = \hbar/\sqrt{2M_{12}}$. The Fourier transformed operator acts on $\phi(x, p_y)$ as
 \begin{equation}\label{four1}
\mathcal{F}_{12} V^{1/2}_{23} \mathcal{F}_{12}^{-1} \phi = \frac 1{(2\pi)^{3/2} \gamma^3}\int d^3 p'_y  \widehat{V^{1/2}_{23}} ((p_y - p'_y)/\gamma) \exp{\Bigl\{i\frac{\beta}{\gamma} x \cdot (p_y - p'_y)\Bigr\}}\phi(x, p'_y),
\end{equation}
where $\widehat{V^{1/2}_{23}} \in L^2(\mathbb{R}^3)$ is a Fourier transform of $V^{1/2}_{23} \in L^2(\mathbb{R}^3)$.
For the kernel of $K_2 (z)$ we get
 \begin{gather}
K_2 (x,p_y;x',p'_y) = \frac 1{2^{7/2}\pi^{5/2} \gamma^3} \left[\frac 1{z + 1 + t(p_y)} - \frac 1{z + 1} \right] V_{12}^{1/2} (\alpha x) \label{four2}\\
\times \frac{e^{-\sqrt{p_y^2 + z^2} |x-x'|}}{|x-x'|}  \exp{\Bigl\{i\frac{\beta}{\gamma} x' \cdot (p_y - p'_y)\Bigr\}} \widehat{V^{1/2}_{23}} ((p_y - p'_y)/\gamma) . \label{four27}
\end{gather}
For the square of the Hilbert-Schmidt norm we obtain
 \begin{equation}\label{four7}
\| K_2 (z) \|^2_2 = \frac 1{2^7 \pi^5 } cc' {\tilde c} \int_{|p_y| \leq 1} d^3 p_y \; \left[\frac 1{z + \sqrt{|p_y|}} - \frac 1{z + 1} \right]^2 \frac 1{\sqrt{p_y^2 + z^2}} ,
\end{equation}
where $c, c'$ are defined in (\ref{aijacts9})--(\ref{aijacts10}) and
\begin{equation}\label{four8}
    {\tilde c} = \frac 1{\gamma^6} \int d^3 p'_y |\widehat{V^{1/2}_{23}} (p'_y/\gamma)|^2
\end{equation}
is finite because $\widehat{V^{1/2}_{23}} \in L^2$.
Estimating the integral in (\ref{four7}) we finally obtain
 \begin{equation}\label{four9}
\| K_2 (z) \|^2_2 \leq \frac 1{2^7 \pi^5 } cc' {\tilde c} \int_{|p_y| \leq 1} d^3 p_y \; \frac 1{p_y^2} = \frac 1{2^5 \pi^4} cc' {\tilde c}  . 
\end{equation}
The strong continuity of $K_2 (z)$ for $z \to +0$ follows from the explicit form of the kernel in (\ref{four2})--(\ref{four27}). The strong continuity of $K_1 (z)$ is proved similarly. \end{proof}

\begin{remark}
Though the operator sequence in Lemma~~\ref{lem:6} is, in fact, norm--continuous for $z \to +0$, the operator sequence in Lemma~\ref{lem:8} is not. To keep the same pattern of proof we prefer to stick to the strong continuity in both cases.
\end{remark}

\begin{lemma}\label{lem:9}
Suppose $H(\lambda)$ defined in (\ref{xc31aa})--(\ref{:xc31aa}) satisfies R1, R4-7. If $\psi_k$ is a weakly convergent subsequence of $\psi_n$, then $V^{1/2}_{ij}\psi_k$ converges in norm.
\end{lemma}
\begin{proof}
We use the functions of the IMS decomposition \cite{ims,teschl}, which for $s=1,2,3$ satisfy $J_s \in C^{2}
(\mathbb{R}^{3N-3})$, $J_s \geq 0$, $\sum_s J^{2}_s =1$ and $J_s (\lambda x) =
J_s (x)$ for $\lambda \geq 1$ and $|x|= 1$. Given these properties of $J_s$ there exists $C > 0$ \cite{ims,teschl} such that for $i \neq s$
\begin{equation}\label{ims5}
    \supp J_s \cap \{ x | |x| > 1 \} \subset \{x|\; |r_i - r_s | \geq C |x|\} .
\end{equation}

By the IMS formula \cite{ims,teschl} the
Hamiltonian $H(\lambda)$ can be decomposed as
\begin{equation}\label{ims}
    H (\lambda) = \sum_{s=1}^3 J_s H_s (\lambda) J_s + K(\lambda) ,
\end{equation}
where
\begin{gather}
H_s (\lambda) = H_0 - \lambda V_{lm}, \quad \quad (l\neq s, m \neq s)\\
K(\lambda) = - \lambda \sum_{s=1}^3  (V_{ls} + V_{ms} )  |J_s |^2 + \sum_{s=1}^3  |\nabla J_s |^2 \label{K} \quad \quad (l\neq s, m \neq s , l \neq m) .
\end{gather}

By condition of the lemma $\psi_k \wto \psi_{cr}$, where $\psi_{cr} \in D(H_0)$ by Lemma~\ref{lem:2}. We shall prove the lemma in three steps given through equations
\begin{gather}
(a) \quad \lim_{k \to \infty} \Bigl((\psi_k - \psi_{cr} ) , K(\lambda_k) (\psi_k - \psi_{cr} ) \Bigr) = 0 \label{kg0} \\
(b) \quad \lim_{k \to \infty} \Bigl((\psi_k - \psi_{cr} )  , H(\lambda_k)  (\psi_k - \psi_{cr} )  \Bigr)= 0 \label{kg0:b}\\
(c) \quad \lim_{k \to \infty} \Bigl((\psi_k - \psi_{cr} )  , V_{ij}  (\psi_k - \psi_{cr} )  \Bigr)= 0 . \label{kg0:c}
\end{gather}
From $(c)$ the statement of the lemma clearly follows. Let us
start with $(a)$. From R6 we have
\begin{equation}\label{tobe}
 |(f,K(\lambda) f)| \leq (f,\tilde K f) \quad \quad (\forall f \in D(H_0) ) ,
\end{equation}
where the operator $\tilde K$ is defined through
\begin{gather}
\tilde K = \lambda \sum_{s=1}^3  (F_{ls} + F_{ms} )  |J_s |^2 + \sum_{s=1}^3  |\nabla J_s |^2 \label{tilk1}  \quad \quad (l\neq s, m \neq s , l \neq m) .
\end{gather}
The first sum in (\ref{tilk1}) is relatively $H_0$ compact by Lemma~\ref{lem:5}, and the second sum is relatively $H_0$ compact because $ |\nabla J_s |^2 \in L^\infty_\infty (\mathbb{R}^{3N-3})$ (see \cite{ims}). Thus  $\tilde K$ is relatively $H_0$ compact and
\begin{equation}\label{j233}
((\psi_k - \psi_{cr} ) , \tilde K (\psi_k - \psi_{cr} ) ) \to 0
\end{equation}
follows from Lemma~\ref{lem:2}.
This proves (a). Rewriting the expression in (b) we obtain
\begin{gather}
\bigl( (\psi_k - \psi_{cr} ), H(\lambda_k) \: (\psi_k - \psi_{cr} )\bigr)
= E(\lambda_k) \bigl((\psi_k - \psi_{cr} ), \psi_k\bigr) - \label{tobe342}\\
\bigl((\psi_k -\psi_{cr} ), H(\lambda_{cr})\psi_{cr} \bigr)  - [\lambda_k - \lambda_{cr}]\bigl((\psi_k -\psi_{cr} ), V\psi_{cr} \bigr) \label{tobe34} ,
\end{gather}
where we have used $H(\lambda_k) = H(\lambda_{cr})
+ [\lambda_k - \lambda_{cr}]V$. All terms on the rhs of (\ref{tobe342})--(\ref{tobe34}) go to zero because $E(\lambda_k) \to 0$ and $\psi_k \wto \psi_{cr}$. It remains to be shown that $(c)$ is true.
\begin{gather}
\lim_{k \to \infty} \Bigl((\psi_k - \psi_{cr} )  , V_{ij}  (\psi_k - \psi_{cr} )  \Bigr)= \sum_{s=1}^3 \lim_{k \to \infty} \Bigl((\psi_k - \psi_{cr} )  , J_s V_{ij} J_s (\psi_k - \psi_{cr} )  \Bigr) \label{sima1}\\
= \lim_{k \to \infty} \Bigl((\psi_k - \psi_{cr} )  , J_l V_{ij} J_l (\psi_k - \psi_{cr} )  \Bigr) \quad \quad (l \neq i \neq j) \label{sima2},
\end{gather}
where we have used that $J_i V_{ij}$ and $J_j V_{ij}$ are relatively $H_0$ compact by Lemma~\ref{lem:5} and the corresponding scalar products vanish by Lemma~\ref{lem:2}.

From (a), (b) and (\ref{ims}) we obtain
\begin{equation}\label{tobe4}
\bigl((\psi_k - \psi_{cr} ), J_l  H_{l} (\lambda_k) J_l (\psi_k - \psi_{cr} )  \bigr)
\to 0 \quad \quad (\forall l) .
\end{equation}
Together with R7 this gives us
\begin{equation}\label{kotz1}
\lim_{k \to \infty} \left((\psi_k - \psi_{cr} )  , J_l V_{ij} J_l (\psi_k - \psi_{cr} )  \right) = 0\quad \quad (l \neq i \neq j) .
\end{equation}
Finally, comparing (\ref{kotz1}) and (\ref{sima1})--(\ref{sima2}) we conclude that (c) holds. \end{proof}

\begin{proof}[Proof of Theorem~\ref{th:2}]
It is enough to show that any weakly converging subsequence of $\psi_n$ converges in norm. Indeed, in this case by Lemma~\ref{lem:4} $\psi_n$ does not spread and thus by Theorem~\ref{th:1} there must exist a bound state at threshold. In order not to overload the notation with additional subscripts we keep the same notation for a weakly converging subsequence, that is we assume $\psi_n \wto \psi_{cr}$ and we must prove $\| \psi_n -\psi_{cr}\| \to 0$.

By Schr\"odinger equation for $k_n^2 = -E_n >0$
\begin{equation}\label{nnr1}
    \psi_n = \lambda_n \sum_{i<j} [H_0 + k_n^2]^{-1} V_{ij} \psi_n = \lambda_n \sum_{i<j} \mathcal{A}_{ij} (k_n)  \bigl[ B^{-1}_{ij} (k_n) V^{1/2}_{ij}\psi_n \bigr] ,
\end{equation}
where $\mathcal{A}_{ij}$ is defined in (\ref{ya23}). By Lemma~\ref{lem:6} $\psi_n$ converges in norm if the sequence $B^{-1}_{ij} (k_n) V^{1/2}_{ij}\psi_n $ does. The convergence of the latter we prove below. From (\ref{nnr1}) we obtain
\begin{equation}\label{nnr2}
    V^{1/2}_{ij} \psi_n = \lambda_n \sum_{l < m} \mathcal{C}_{ij;lm} (k_n) [V^{1/2}_{lm} \psi_n ] .
\end{equation}
Using (\ref{ay2222}) we rewrite (\ref{nnr2})
\begin{equation}\label{nnr4}
    V^{1/2}_{ij} \psi_n = \lambda_n \mathcal{R}_{ij} (k_n) \sum_{\substack{l < m \\
    (lm) \neq (ij)}} \mathcal{C}_{ij;lm}(k_n) (V^{1/2}_{lm} \psi_n ) .
\end{equation}
Now we act with $B^{-1}_{ij}(k_n)$ on both parts of (\ref{nnr4}) and use that it commutes with $\mathcal{R}_{ij} (k_n)$
\begin{equation}\label{nnr5}
    B^{-1}_{ij}(k_n) V^{1/2}_{ij} \psi_n = \lambda_n \mathcal{R}_{ij} (k_n) \sum_{\substack{l < m \\
    (lm) \neq (ij)}} B^{-1}_{ij}(k_n) \mathcal{C}_{ij;lm}(k_n) (V^{1/2}_{lm} \psi_n )   . 
\end{equation}
By Lemmas~\ref{lem:7},\ref{lem:8},\ref{lem:9} the rhs converges in norm.\end{proof}

\begin{acknowledgements}
The author would like to thank Prof. Walter Greiner for the warm hospitality at FIAS.
\end{acknowledgements}


\begin{thebibliography}{99}

\bibitem{klaus} M.~Klaus and B.~Simon, Ann.~Phys. (N.Y.) {\bf 130}, 251
(1980)


\bibitem{klaus2} M.~Klaus and B.~Simon, Comm.~Math.~Phys.~{\bf 78},
153 (1980)

\bibitem{simon} B.~Simon, J.~Functional Analysis~{\bf 25}, 338
(1977)


\bibitem{ostenhof} M. Hoffmann-Ostenhof, T. Hoffmann-Ostenhof and
B. Simon, J.~Phys.~A~{\bf 16}, 1125 (1983)


\bibitem{gest} D.~Bolle, F.~Gesztesy and W.Schweiger,
J.~Math.~Phys~{\bf 26}, 1661 (1985)




\bibitem{prl} D.~K. Gridnev and M.~E.~Garcia, J.~Phys.~A: Math. Theor. \textbf{40}, 9003  (2007)

\bibitem{karner} G. Karner, Few-Body Systems \textbf{3}, 7 (1987)

\bibitem{submit} D. K. Gridnev, arXiv:arXiv:0912.0418v2

\bibitem{vaagen}
M.V. Zhukov, B.V. Danilin, D.V. Fedorov, J.M. Bang, I.J. Thompson
and J.S. Vaagen, Phys.~Rep. {\bf 231}, 151 (1993).

\bibitem{fedorov}
A.~S.~Jensen, K.~Riisager, and D.~V.~Fedorov,
Rev.~Mod.~Phys.~{\bf 76} 215 (2004); K.~Riisager, D.~V.~Fedorov and
A.~S.~Jensen, Europhys.~Lett.~{\bf 49}, 547 (2000).

\bibitem{efimov}
T. Kraemer, {\em et.al.} Nature {\bf 440}, 315 (2006)


\bibitem{zhislin} G.~M.~Zhislin, Trudy~Mosk.~Mat.~Ob{\v
s}{\v c}. {\bf 9}, 81 (1960); E.~F.~Zhizhenkova and G.~M.~Zhislin,
Trudy~Mosk.~Mat.~Ob{\v s}{\v c}. {\bf 9}, 121 (1960)

\bibitem{loss} E. H. Lieb and M. Loss, {\em Analysis}, AMS (1997)

\bibitem{reed} M. Reed and B. Simon, {\em Methods of Modern
Mathematical Physics}, vol.~4, Academic Press/New York (1978)

\bibitem{teschl} G. Teschl, {\em Mathematical Methods in Quantum
Mechanics; With Applications to Schr\"odinger Operators}, Lecture
Notes (2005),
http://www.mat.univie.ac.at/~gerald/ftp/book-schroe/index.html

\bibitem{faddeev} L. D. Faddeev, Trudy Mat. Inst. Steklov. \textbf{69} (1963) (Russian)

\bibitem{sobolev} A. V. Sobolev, Commun. Math. Phys. \textbf{156}, 101 (1993)

\bibitem{yafaev} D. R. Yafaev, Math. USSR-Sb. \textbf{23}, 535 (1974); Notes of LOMI Seminars \textbf{51} (1975) (Russian)

\bibitem{greiner} W. Greiner, {\em Quantum Mechanics: An
Introduction}, Springer--Verlag, Berlin (2000)














\bibitem{ims} H.~L. Cycon, R.~G. Froese, W. Kirsch and B. Simon,
{\em Schr\"odinger Operators with Applications to Quantum Mechanics
and Global Geometry}, Springer--Verlag, Berlin Heidelberg (1987)




%
%
%
%
%
%






%
%


\end{thebibliography}
\end{document}